%% file: main.tex
\documentclass{article}
\input{preamble}\usepackage{capt-of}
\input{preprint}\title{CascadeEP: Asynchronous Expert Execution\\
for MoE Prefill under Attention Imbalance}
\input{authors}%
\input{commands}

\begin{document}
\maketitle
\authornotes
\input{sections/abstract}%


\input{sections/introduction}
\FloatBarrier
\input{sections/related_work}%

\input{sections/design}
\FloatBarrier
\input{sections/evaluation}%

\input{sections/conclusion}%

\input{sections/reproducibility_statement}%

\input{sections/ai_statement}%

\FloatBarrier
\bibliography{references}
\bibliographystyle{iclr2027_shortauthor}
\appendix
\clearpage
\renewcommand{\thefigure}{\thesection\arabic{figure}}%
\renewcommand{\theHfigure}{\thesection.\arabic{figure}}%
\renewcommand{\thetable}{\thesection\arabic{table}}%
\renewcommand{\theHtable}{\thesection.\arabic{table}}%
\counterwithin*{figure}{section}%
\counterwithin*{table}{section}%
\input{appendices/dep_prefill}%

\input{appendices/dsv4_attention}%

\input{appendices/runtime_implementation}%

\input{appendices/estimate_accuracy}%

\input{appendices/analysis}\end{document}

%% file: preamble.tex
\usepackage{iclr2027_conference,times}
\usepackage[T1]{fontenc}
\usepackage{amsmath,amssymb,amsthm}
\usepackage{algorithm,algpseudocode}
\usepackage{booktabs}
\usepackage{enumitem}
\usepackage{graphicx}
\usepackage{wrapfig,needspace}
\restylefloat{algorithm}
\usepackage{xcolor}
\usepackage{placeins}
\usepackage{hyperref}
\definecolor{figureref}{HTML}{15803D}
\definecolor{planningnote}{HTML}{1F4E79}
\hypersetup{colorlinks=true,allcolors=black,linkcolor=figureref}
\usepackage{url}

%% file: preprint.tex
\iclrfinalcopy
\renewcommand{\headrulewidth}{0pt}
\makeatletter
\renewcommand{\@maketitle}{\vbox{\hsize\textwidth
  \fancyhead{}
  {\centering\LARGE\bfseries\@title\par}
  \vskip 0.18in
  {\centering\normalfont\@author\par}
  \vskip 0.18in
}}
\makeatother
\hypersetup{
  pdftitle={CascadeEP: Asynchronous Expert Execution for MoE Prefill under Attention Imbalance},
  pdfauthor={Jin Qin, Tiancheng Hu, Shiyan Wang, Junhao Hu, Zexin Jian, Yuzheng Wang, Haoyu Li, Chunwei Xia, Ying Liu, Pixian Zhan, Di Wang, Zhongzhe Hu, Huimin Cui, Tao Xie, Chenxi Wang}
}

%% file: authors.tex
\author{Tiancheng Hu}
\author{%
\begin{minipage}{\textwidth}
\centering\small
Jin Qin\textsuperscript{1,2,*}\quad Tiancheng Hu\textsuperscript{3,*}\quad Shiyan Wang\textsuperscript{4}\quad Junhao Hu\textsuperscript{3}\quad Zexin Jian\textsuperscript{2}\\[3pt]
Yuzheng Wang\textsuperscript{3}\quad Haoyu Li\textsuperscript{1}\quad Chunwei Xia\textsuperscript{5}\quad Ying Liu\textsuperscript{2}\quad Pixian Zhan\textsuperscript{6}\\[3pt]
Di Wang\textsuperscript{3}\quad Zhongzhe Hu\textsuperscript{7}\quad Huimin Cui\textsuperscript{2}\quad Tao Xie\textsuperscript{3, 8,\textdagger}\quad Chenxi Wang\textsuperscript{2,\textdagger}\\[7pt]
\footnotesize
\textsuperscript{1}University of the Chinese Academy of Sciences\\[2pt]
\textsuperscript{2}Institute of Computing Technology, Chinese Academy of Sciences\\[2pt]
\textsuperscript{3}Peking University\quad \textsuperscript{4}Beijing University of Posts and Telecommunications\\[2pt]
\textsuperscript{5}University of Leeds\quad \textsuperscript{6}Advanced Institute of Information Technology\quad \textsuperscript{7}Huawei Technologies Ltd.\\[2pt]
\textsuperscript{8}Beijing Tongming Lake Information Technology Application Innovation Center
\end{minipage}%
}
\newcommand{\authornotes}{%
  \begingroup
  \renewcommand{\thefootnote}{\fnsymbol{footnote}}%
  \footnotetext[1]{Equal contribution.}%
  \footnotetext[2]{Corresponding authors.}%
  \endgroup
}

%% file: commands.tex
\newcommand{\figref}[2][]{\textcolor{figureref}{\ref{#2}#1}}
\newcommand{\secref}[2][]{\textcolor{figureref}{\S\ref{#2}#1}}
\newcommand{\name}{\textsc{CascadeEP}}
\newtheorem{proposition}{Proposition}

%% file: sections/abstract.tex
\begin{abstract}
Mixture-of-experts (MoE) serving commonly deploys data and expert parallelism (DEP): attention replicas run distinct request batches while routed experts are sharded across an expert-parallel (EP) group. During prefill, attention replicas finish dispatch at different times, but synchronous EP delays expert feed-forward network (FFN) computation until routed inputs from all replicas are ready. Request schedulers seek to balance load while reusing the key-value (KV) cache of shared prompt prefixes to avoid redundant prefill computation. These goals can conflict when a replica holding a matching prefix is already overloaded, leaving residual attention imbalance. We present \name{}, a distributed execution engine for MoE prefill. \name{} proposes three mechanisms. Asynchronous EP allows expert computation to start before tokens from all attention replicas are ready. \textit{streamFFN} batches ready tokens to balance early execution with FFN computation efficiency. Opportunistic expert weight fetching (\textit{OEWF}) allows a faster replica to fetch expert weights and execute unstarted work from other GPUs. We evaluate \name{} on DeepSeek-V4-Flash, DeepSeek-V4-Pro, and GLM-5.3, and our results show that \name{} achieves up to  1.48$\times$ speedup in p95 time-to-first-token (TTFT) and improves the inference throughput by up to 1.17$\times$.
\end{abstract}

%% file: sections/introduction.tex
\section{Introduction}
\label{sec:introduction}

Mixture-of-experts (MoE) has become a mainstream architecture for scaling large language models, as exemplified by DeepSeek-V4 and Kimi K3 \citep{deepseek2026v4,kimiteam2026kimik3openfrontier}, which expand parameter capacity while activating only a subset of experts per token. For large-scale MoE serving, a widely adopted deployment strategy is data and expert parallelism (DEP): attention uses data parallelism (DP), while routed experts use expert parallelism (EP) \citep{vllm2026agentx}. Attention replicas process distinct request batches and maintain the associated key-value (KV) caches, while routed experts are distributed across GPUs within an EP group. Token dispatch transfers token activations---the inputs to routed experts---to the GPUs hosting the selected experts for feed-forward network (FFN) computation. Token combine returns expert outputs to the originating attention replicas and aggregates them \citep{deepep2025}. Batching routed tokens from multiple attention replicas can increase per-expert batch sizes and improve expert GEMM efficiency \citep{deepseek2025inference}. This deployment can also integrate other parallelism techniques, such as tensor parallelism (TP) within each attention replica \citep{zhu2025megascaleinfer}.

\textbf{DP attention imbalance in MoE prefill.} 
During MoE prefill, attention replicas process different numbers of requests with varying context lengths, causing them to complete attention computation and reach EP dispatch at different times \citep{deepseek2025inference,zhu2025megascaleinfer}. Existing DEP implementations typically use synchronous EP, in which each GPU starts FFN computation only after all attention replicas have finished dispatching their tokens \citep{deepep2025,moonep2026}. Consequently, replicas that finish attention early remain idle before FFN computation while waiting for other replicas to reach dispatch, as illustrated in Figure~\figref[(a)]{fig:overview}. In our DeepSeek-V4-Flash \citep{deepseek2026v4} prefill experiment on the AgentX dataset \citep{semianalysis2026agentx} with DP8+EP8 in SGLang \citep{zheng2024sglang}, replicas spend an average of 28.86\% of end-to-end runtime idle while waiting for other replicas to reach dispatch (Appendix~\ref{app:dep-prefill}).

\begin{figure}[]
\centering
\includegraphics[width=\linewidth]{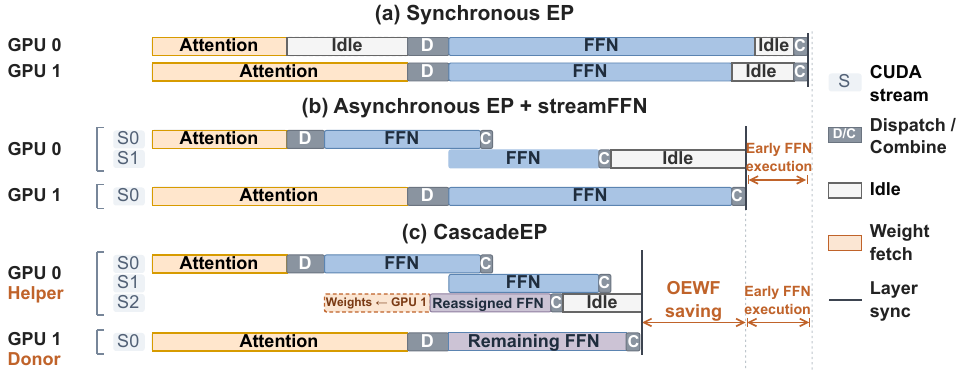}
\caption{Illustrative MoE prefill timelines for two representative GPU ranks. (a) Synchronous EP synchronizes all replicas before starting FFN computation. (b) Asynchronous EP with \textit{streamFFN} overlaps independent FFN GEMM kernels on multiple CUDA streams. (c) \name{} further allows GPU~0 to fetch weights and execute reassigned FFN work on a dedicated helper stream. }
\label{fig:overview}
\end{figure}


\Needspace{28\baselineskip}
\begin{wrapfigure}{r}{0.49\textwidth}
\centering
\includegraphics[width=\linewidth]{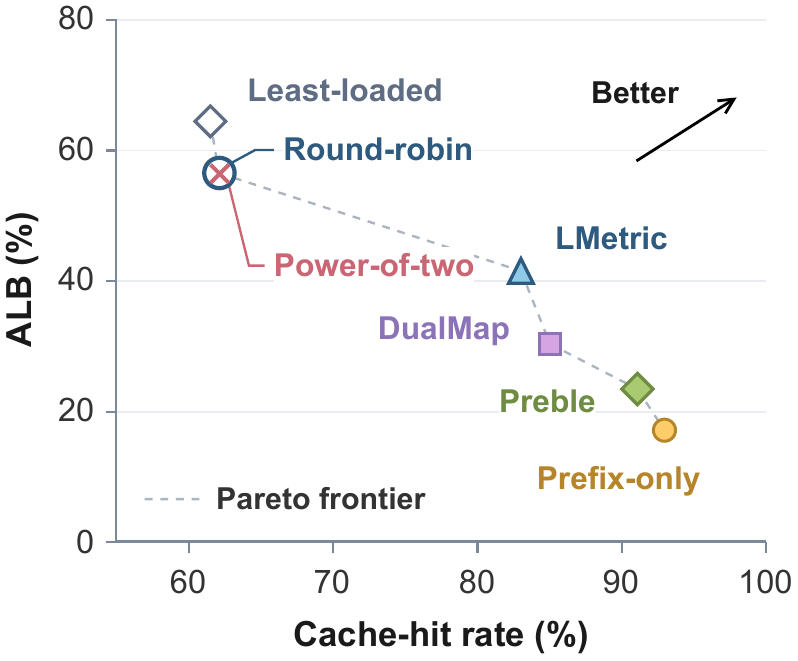}
\caption{Cache reuse vs. ALB for state-of-the-art scheduling policies, evaluated with SGLang serving DeepSeek-V4-Flash during prefill on the AgentX dataset. ALB is computed for each layer as the mean attention-stage duration across DP replicas divided by the maximum.}
\label{fig:motivation}
\end{wrapfigure}

\textbf{Limits of request scheduling.}
Existing request scheduling policies balance workloads across serving instances while reusing KV caches for matched request prefixes \citep{srivatsa2025preble,yuan2026dualmap,zhang2026simple}, which can also be adapted to DEP's attention replicas. However, request scheduling policies alone cannot eliminate the imbalance in DP attention workloads for two reasons. First, cache reuse constrains request scheduling: a replica holding a matched prefix may be overloaded, whereas choosing a less loaded replica without the matched prefix may require KV recomputation or additional transfer \citep{qin2025mooncake}. This creates a trade-off between cache-hit rate and attention load balance (ALB), as shown in Figure~\figref{fig:motivation}. Second, even setting aside cache-locality constraints, attention-stage execution time depends on input sequence lengths and batch composition \citep{zhong2024distserve,zhu2025megascaleinfer}. Unpredictable online arrivals make it difficult to consistently balance attention-stage execution times across replicas (see Appendix~\ref{app:profiles}).

\textbf{Design principles: asynchronous EP with expert weight fetching.} To address DP attention imbalance that persists under state-of-the-art request scheduling policies, we propose \name{}, a distributed execution engine for MoE prefill that combines asynchronous expert execution with expert weight fetching. \name{} allows FFN GEMM kernels to be launched before all routed expert tokens have arrived and to return output asynchronously (Figure~\figref[(b)]{fig:overview}). \name{} also enables replicas that have completed dispatch to fetch expert weights from those still computing attention and take over part of their pending FFN work. This reduces the remaining FFN workload on replicas with longer attention times, helping mitigate the layer-wise straggler effect (Figure~\figref[(c)]{fig:overview}).

\textbf{Challenge 1: preserving FFN computation efficiency.} 
Under synchronous EP, each GPU waits for tokens from all attention replicas before launching FFN GEMM kernels. Asynchronous EP allows FFN to start on ready expert tokens but splits a complete GEMM kernel into smaller ones, reducing batch sizes and potentially degrading computation efficiency. To preserve FFN computation efficiency, we introduce \textit{streamFFN}, guided by a key observation: as the number of expert tokens increases, effective FFN throughput generally rises rapidly at first, then improves more gradually before approaching a plateau. Accordingly, \textit{streamFFN} accumulates ready expert tokens until their count reaches a launch threshold. Meanwhile, \textit{streamFFN} schedules independent FFN GEMM kernels for different input groups on multiple CUDA streams, allowing concurrent kernels to use GPU resources left idle by partially occupied final thread-block waves \citep{nvidia2026matmul,nvidia2026concurrency}.


\textbf{Challenge 2: utilizing idle GPU time.} 
Although asynchronous EP allows replicas that finish attention early to complete their FFNs sooner, the next attention layer cannot begin until all replicas have received their required expert outputs. This layer-wise barrier leaves earlier-finishing replicas idle while waiting for the slowest. To exploit this idle time, we introduce \emph{Opportunistic Expert Weight Fetching} (\textit{OEWF}), which enables replicas that have completed token dispatch to asynchronously fetch expert weights from those still computing attention and take over part of their unstarted FFN work. \textit{OEWF} first uses a coordinated greedy algorithm to select the source replicas and expert weights to fetch. Once the weights are ready, a cost model compares the completion times of the current and revised FFN execution plans to determine whether to reassign FFN work.

\textbf{Results.} We evaluate \name{} on DeepSeek-V4-Flash, DeepSeek-V4-Pro \citep{deepseek2026v4}, and GLM-5.3 \citep{zai2026glm53} using Tool\&Agent \citep{qin2025mooncake} and AgentX under five request-scheduling policies.
End-to-end experiments show that \name{} achieves 1.13$\times$ speedup in p95 time-to-first-token (TTFT) on average (up to 1.48$\times$) and improves inference throughput by 1.06$\times$ on average (up to 1.17$\times$).

\noindent Our contributions in this paper are as follows:
\begin{itemize}[leftmargin=*,labelindent=0pt,itemindent=0pt,listparindent=0pt]
\item We characterize DP attention imbalance in MoE prefill and show how synchronous EP converts uneven attention completion times into GPU idle time.
\item We propose \name{}, a distributed MoE prefill engine that combines asynchronous EP, \textit{streamFFN}, and \textit{OEWF} to execute ready tokens efficiently and reassign FFN work to idle GPUs.
\item We comprehensively evaluate \name{} across three models, two real-world workloads, and five request-placement policies. Our experiments show \name{} achieves significant end-to-end improvement compared with baselines.
\end{itemize}

%% file: sections/related_work.tex
\section{Related Work}
\label{sec:relatedwork}
\paragraph{EP communication optimizations.}
In large-scale DEP serving, dispatch and combine incur substantial communication overhead from large data volumes and uneven traffic across GPUs. Existing works mitigate this overhead through two complementary approaches. \emph{Communication--computation overlap} pipelines token transfers with expert computation: DeepEP provides specialized dispatch and combine kernels with asynchronous interfaces \citep{deepep2025}, while Comet schedules communication and GEMM execution at tile granularity within fused GPU kernels \citep{zhang2025comet}. \emph{Fused MoE execution} reduces kernel-launch and orchestration overhead: FlashMoE integrates dispatch, expert computation, and combine into a single persistent GPU kernel, pipelining them via device-initiated communication and task scheduling \citep{aimuyo2025flashmoe}. These operator optimizations, together with expert placement and load balancing \citep{li2023lina,yang2026libra}, target communication and expert execution within the MoE layer. They are largely orthogonal to our focus on synchronization stalls caused by DP attention imbalance.

\paragraph{Request scheduling for DP load balancing.}
Round-robin assigns requests cyclically across instances \citep{sglang2026dp}, while prefix-only routing selects the instance with the longest matching cached prefix \citep{srivatsa2025preble}. Power-of-two-choices selects the less loaded of two sampled instances \citep{mitzenmacher2001poweroftwo}, whereas least-loaded routing selects the least loaded overall \citep{zhang2026simple}. LMetric minimizes the product of the cache-aware queued prefill-token count after assignment and the current batch size \citep{zhang2026simple}. Preble's global E2 scheduler combines prefix reuse with load-aware exploration based on computation and KV-cache eviction costs \citep{srivatsa2025preble}. DualMap uses independent prefix hashes to identify two candidates, then applies SLO-aware routing and hotspot rebalancing \citep{yuan2026dualmap}.

%% file: sections/design.tex
\section{\name{} Design}
\label{sec:design}
\subsection{Overview}
\label{sec:execution-model}
\begin{figure}[!t]
\centering
\includegraphics[width=0.9\linewidth]{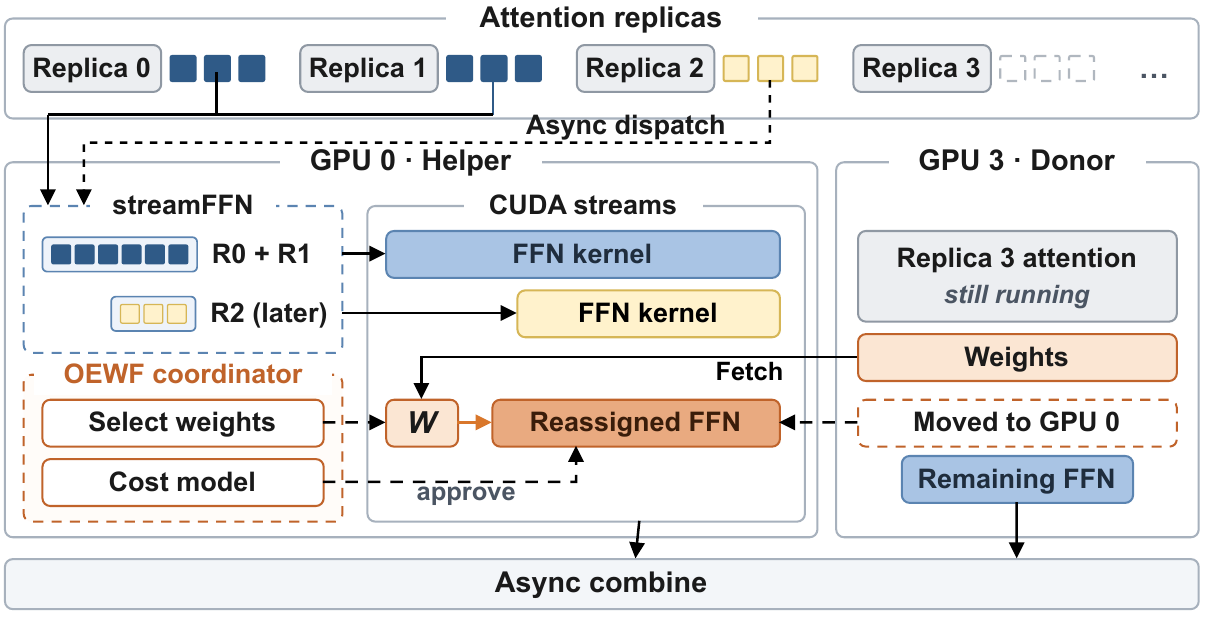}
\caption{\name{} system overview.}
\label{fig:system}
\end{figure}

\name{} is a distributed execution engine for MoE prefill. Figure~\figref{fig:system} shows its main components: asynchronous dispatch and combine for communication, \textit{streamFFN} (\secref{sec:sgffn}) for FFN GEMM kernel batching and scheduling, and opportunistic expert weight fetching (\textit{OEWF}, \secref{sec:swf}) for fetching expert weights and reassigning FFN work.
Each replica dispatches tokens independently after completing its attention stage, and \textit{streamFFN} batches ready tokens for FFN execution. A replica that has completed token dispatch can act as a \emph{helper}. For each helper, the coordinator selects a replica still computing attention as its \emph{donor} and determines which expert weights to fetch. After the helper finishes fetching the selected experts' full weights, the coordinator uses a cost model to decide whether to reassign the corresponding FFN work from the donor to the helper.

\subsection{Asynchronous EP and streaming FFN}
\label{sec:sgffn}

\name{} enables FFN to begin before all replicas complete token dispatch. Our extension to DeepEP \citep{deepep2025} exposes per-replica completion signals at each destination, indicating when an attention replica's routed expert tokens have arrived and are ready for computation. These tokens become eligible for \textit{streamFFN} scheduling without waiting for other replicas. Appendix~\ref{app:runtime-impl} details our implementation in SGLang \citep{zheng2024sglang}.

However, launching a separate FFN GEMM kernel whenever a replica's routed tokens become ready can fragment the workload into small per-expert batches and reduce computation efficiency. \textit{streamFFN} balances early execution with FFN efficiency through two mechanisms: threshold-based batching to form larger FFN GEMM kernels from ready expert tokens, and multi-stream execution to overlap independent kernels and utilize otherwise idle GPU resources.

\paragraph{Launch threshold.}
Our key observation is that FFN throughput rises rapidly at low expert-token counts but gradually saturates as the token count increases (Figure~\figref{fig:ffn-profile}). This motivates starting FFN computation at near-peak efficiency without waiting for all tokens to arrive. Accordingly, \textit{streamFFN} sets the launch threshold $\theta$ to the smallest profiled expert-token count that achieves at least a target fraction $\rho$ of the observed peak throughput. A larger $\rho$ favors computational efficiency, whereas a smaller $\rho$ enables earlier launches. The threshold is calibrated for each model--GPU configuration to account for expert shapes, compute precision, and kernel tiling. Once all tokens have arrived, \textit{streamFFN} bypasses the threshold and processes the remaining expert tokens.

\paragraph{Multi-stream execution.}
A launch threshold improves FFN computation efficiency but does not eliminate \emph{wave quantization}, also known as the \emph{tail effect} \citep{nvidia2026matmul}. The number of blocks in a full wave equals the SM count multiplied by the kernel's resident-block capacity per SM. When a kernel's remaining block count is not a multiple of this capacity, its final wave is only partially occupied. \textit{streamFFN} schedules independent FFN GEMM kernels for different input groups on separate CUDA streams, allowing eligible kernels to use GPU resources left idle by another kernel's final wave. Our experiments show that \textit{streamFFN} incurs about a 2--3\% loss in effective FFN throughput relative to synchronous execution, while substantially advancing FFN completion through earlier launches, as detailed in \secref{sec:ffn-efficiency}.

\begin{figure}[!t]
\centering
\begin{minipage}[t]{0.46\linewidth}
\vspace{0pt}
\centering
\includegraphics[width=\linewidth]{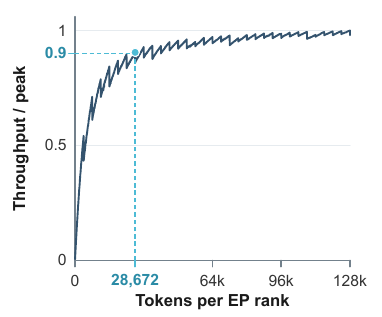}
\caption{FFN throughput of DeepSeek-V4-Flash with EP8 using DeepGEMM on B200, normalized to its observed peak.}
\label{fig:ffn-profile}
\end{minipage}\hfill
\begin{minipage}[t]{0.52\linewidth}
\vspace{0pt}
\setlength{\intextsep}{0pt}
\begin{algorithm}[H]
\small
\caption{Greedy expert-weight selection}
\label{alg:weight-fetch}
\begin{algorithmic}[1]
\Procedure{SelectWeights}{$h$}
  \State $\mathcal D\gets\Call{CandidateDonors}{S,h}$
  \State \textbf{if} $\mathcal D=\varnothing$ \textbf{then} \Return
  \State $d^*\gets\arg\max_{d\in\mathcal D}T_{d}(S)$
  \State $\mathcal L\gets\operatorname{sort}^{\downarrow}_{e\in\mathcal E_{d^*}}\bigl(|\mathcal T_{h{d^*}e}|/b_e\bigr)$
  \State $\mathcal E_h\gets\varnothing$; $R\gets B_h$
  \For{$e\in\mathcal L$}
    \If{$b_e\le R$}
      \State $\mathcal E_h\gets\mathcal E_h\cup\{e\}$; $R\gets R-b_e$
    \Else
      \State \textbf{break}
    \EndIf
  \EndFor
  \If{$\mathcal E_h\ne\varnothing$}
    \State $P_h\gets(h,d^*,\mathcal E_h,\{\mathcal T_{h{d^*}e}\}_{e\in\mathcal E_h})$
    \State $\Call{RecordFetch}{P_h}$
    \State $\Call{AsyncFetch}{P_h}$
  \EndIf
\EndProcedure
\end{algorithmic}
\end{algorithm}
\end{minipage}
\end{figure}

\Needspace{5\baselineskip}
\subsection{Opportunistic Expert Weight Fetching}
\label{sec:swf}
Although asynchronous EP allows replicas that finish attention early to complete their FFN computation sooner, the next attention layer cannot begin until all replicas have received their required expert outputs. This layer-wise barrier leaves earlier-finishing replicas idle while waiting for the slowest. \textit{OEWF} exploits this idle GPU time to help slower replicas execute their unstarted FFN work. We call a replica that has completed its token dispatch a \emph{helper}, while a slower replica still computing attention is a candidate \emph{donor}, from which the helper may fetch expert weights. Each replica will make two decisions: a coordinated greedy algorithm selects a donor and the expert weights to fetch, and once the weights are ready, a cost model decides whether the helper should take over the corresponding FFN work.

\paragraph{Coordinated greedy expert-weight selection.}
Starting from the second MoE layer, each helper makes one fetch decision immediately after dispatch. Candidate donors are replicas still computing attention, and $\mathcal D$ denotes this candidate set.

We first estimate the remaining attention time. For donor replica $d \in\mathcal D$ in the current MoE layer, let $I_d$ and $P_d$ denote its current input lengths and matched prefix lengths, and let $t_d^{\rm att}$ be the elapsed attention time. The remaining attention-stage time is estimated as
\begin{equation}
 A_{d}
= \bigl[f(I_d, P_d)
- t_{d}^{\mathrm{att}}\bigr].
\label{eq:attention-remaining}
\end{equation}
where $f$ predicts the complete attention-stage duration from the input lengths (Appendix~\ref{app:profiles}).

Second, we model dispatch and FFN GEMM kernel latencies. Let $D_{d}$ denote the dispatch latency for rank $d$ in the current layer, which is calculated by dividing the dispatched-token count, determined from the input lengths, by the profiled effective dispatch throughput. We estimate the received expert-token counts $N_{d}$ from the preceding layer's per-rank token count and model FFN latency as
\begin{equation}
F(N_{d}) = \frac{N_{d}}{\eta(N_{d})P_{\mathrm{peak}}}
\label{eq:ffn-throughput-latency}
\end{equation}
where $\eta(N_{d})\in(0,1]$ is the normalized effective FFN throughput and $P_{\mathrm{peak}}$ is the corresponding profile's observed peak throughput.

Finally, we combine these estimates under the current execution plan $S$, and the predicted completion time of donor replica $d$ is
\begin{equation}
T_{d}(S)
= A_{d} + D_{d}
+ F(N_{d})
\label{eq:donor-completion}
\end{equation}

The coordinator selects the donor with the latest completion time:
\begin{equation}
d^*=\arg\max_{d\in\mathcal D}T_{d}(S)
\label{eq:donor-selection}
\end{equation}

For the selected donor $d^*$, let $\mathcal T_{h{d^*}e}$ denote the set of tokens dispatched by helper $h$ to expert $e$, and let $b_e$ be the number of bytes required to fetch that expert's weights. Experts in the eligible set $\mathcal E_{d^*}=\{e:\mathcal T_{h{d^*}e}\ne\varnothing\}$ are sorted in descending order of $|\mathcal T_{h{d^*}e}|/b_e$. The per-helper byte budget $B_h$ is capped by the available capacity of the helper's temporary weight buffer. Starting with $R=B_h$, the coordinator adds each expert to the selected set $\mathcal E_h$ if $b_e\le R$, then reduces $R$ by $b_e$, until adding the next expert would exceed the capacity limit. The coordinator records each fetch plan $P_h$ and its resource usage to update the FFN execution plan, keeping expert selection consistent with a shared global view across all replicas. Finally, the weight transfers are performed asynchronously without interfering with the helper's computation. Algorithm~\ref{alg:weight-fetch} details the procedure.

\paragraph{Cost-based FFN reassignment.}
After helper $h$ finishes fetching the selected weights, the coordinator compares the current plan $S$ with a revised plan $S'$ that transfers $N_{trans}$ unstarted expert tokens from donor $d$ to the helper $h$. The transferred work executes in a separate FFN GEMM kernel on another CUDA stream and can overlap with the helper's remaining FFN work. The completion times of the revised plan are
\begin{equation}
\begin{gathered}
T_d(S')
= A_d + D_d + F\!\left(N_d-N_{\mathrm{trans}}\right)\\
T_h(S')
= T_h(S) + \alpha_h F(N_{\mathrm{trans}})
\end{gathered}
\label{eq:swf-completion-times}
\end{equation}
Here, $\alpha_h\in[0,1]$ is the fraction of the reassigned FFN latency that extends beyond that time. The coordinator assigns $\alpha_h$ from the candidate CUDA-stream schedule under the current plan $S$: $\alpha_h=1$ corresponds to serial execution, while $\alpha_h=0$ means that the reassigned work is fully hidden within the helper's remaining execution. The revised plan is adopted only if
\begin{equation}
\max\{T_d(S'),T_h(S')\}
\le
\max\{T_d(S),T_h(S)\}.
\label{eq:swf-admission}
\end{equation}

Under our cost-model and execution assumptions, we prove that applying \textit{OEWF}'s policy across all replicas guarantees $M_{\rm OEWF}\le M_{\rm no\text{-}fetch}$, where $M$ denotes layer completion time (Appendix~\ref{app:fetch-proof}). Its practical effectiveness depends on these predictions. The complete attention-stage duration and the received expert-token count have accuracies of 90.3\% and 85.9\%, respectively (Appendix~\ref{app:estimate-accuracy}). Section~\ref{sec:oewf-layer} further evaluates its layer-level benefits in real model deployment.

%% file: sections/evaluation.tex
\section{Evaluation Results}
\label{sec:evaluation}

\begin{figure}[!t]
\centering
\begin{minipage}[t]{0.34\linewidth}
    \vspace{0pt}
    \centering
    \includegraphics[width=\linewidth]{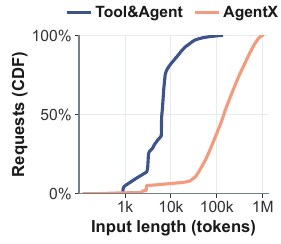}
    \caption{Input-length CDFs of two datasets.}
    \label{fig:workload-lengths}
\end{minipage}\hfill
\begin{minipage}[t]{0.64\linewidth}
    \vspace{0pt}
    \centering
    \footnotesize
    \setlength{\tabcolsep}{1.1pt}
    \renewcommand{\arraystretch}{1.35}

    \begin{minipage}[t][120.8pt][t]{\linewidth}
        \centering
        \begin{tabular}{@{}lccc@{}}
            \toprule
             & \shortstack{\textbf{DeepSeek-V4-}\\\textbf{Flash}}
             & \shortstack{\textbf{DeepSeek-V4-}\\\textbf{Pro}}
             & \raisebox{5.5pt}{\textbf{GLM-5.3}} \\
            \midrule
            \textbf{Total / active parameters}
                & 284B / 13B & 1.6T / 49B & 753B / ${\sim}$40B \\
            \textbf{Attention}
                & CSA and HCA & CSA and HCA & DSA \\
            \textbf{Routed experts / layer}
                & 256 & 384 & 256 \\
            \textbf{Shared experts / layer}
                & 1 & 1 & 1 \\
            \textbf{Routed experts / token}
                & 6 & 6 & 8 \\
            \textbf{Expert intermediate dim.}
                & 2,048 & 3,072 & 2,048 \\
            \bottomrule
        \end{tabular}
    \end{minipage}

    \captionof{table}{Model architecture details.
    GLM-5.3 parameter counts follow \citet{nvidia2026glm53}.}
    \label{tab:eval-models}
\end{minipage}
\end{figure}

\subsection{Evaluation Settings}
\label{sec:setup}
\paragraph{Workloads.}
We evaluate two workloads separately. Mooncake's Tool\&Agent trace captures production tool and agent requests characterized by long, repeated system prompts \citep{qin2025mooncake}. The AgentX trace corpus contains multi-turn coding-agent sessions, including both main-session and subagent requests \citep{semianalysis2026agentx}. For each workload, we scale the original interarrival times to match the target request rate. For each model and target rate, all baselines replay the same request trace under the same KV-cache budget. Figure~\figref{fig:workload-lengths} shows the input-length distributions.

\paragraph{Models and environments.}
We evaluate \name{} on DeepSeek-V4-Flash, DeepSeek-V4-Pro \citep{deepseek2026v4}, and GLM-5.3 \citep{zai2026glm53}. For all three models, we use prefill--decode (PD) disaggregation across two nodes and evaluate only the prefill instance. The prefill node contains eight NVIDIA B200 GPUs interconnected via NVLink, with eight-way data parallelism for attention and eight-way expert parallelism for routed experts, placing one attention replica and one EP rank on each GPU. Attention uses FP8 for all three models; routed-expert weights use MXFP4 for DeepSeek-V4-Flash and DeepSeek-V4-Pro, and FP8 for GLM-5.3. We set \textit{streamFFN}'s throughput target to $\rho=0.90$ for all three models. Table~\ref{tab:eval-models} summarizes the model architectures.


\begin{figure}[]
\centering
\includegraphics[width=\linewidth]{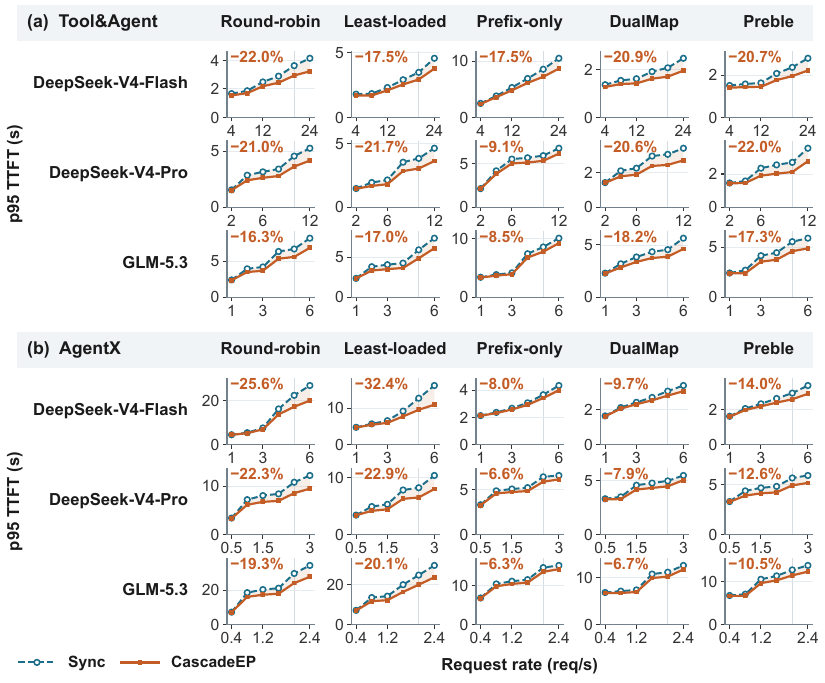}
\caption{Prefill p95 TTFT across models, workloads, and request scheduling policies. Percentages denote reductions at the highest request rates.}
\label{fig:e2e}
\end{figure}

\paragraph{Baseline and placement policies.}
Our baseline, \textsc{Sync}, uses synchronous SGLang DEP deployment with DeepEP for dispatch and combine \citep{deepep2025} and DeepGEMM for FFN GEMM computation. FFN starts after dispatch completes on all replicas (Figure~\figref[(a)]{fig:overview}). We compare \textsc{Sync} and \name{} under five request-scheduling policies: round-robin (SGLang default), least-loaded, prefix-only, DualMap, and Preble (\secref{sec:relatedwork}) \citep{sglang2026dp,zhang2026simple,srivatsa2025preble,yuan2026dualmap}. We implement these policies in the SGLang router to assign requests across DP attention replicas within the prefill instance. For Preble, we use a prefill-oriented adaptation of its global E2 policy. We use round-robin placement for the ablation experiments.

\paragraph{Metrics.}
We report p95 time-to-first-token (TTFT) across request rates. We also report prompt-token throughput at selected high-load request rates near saturation. Each point in the end-to-end and component-ablation results is the arithmetic mean of the corresponding per-run metric over five independent runs. 

\begin{figure}[]
\centering
\includegraphics[width=\linewidth]{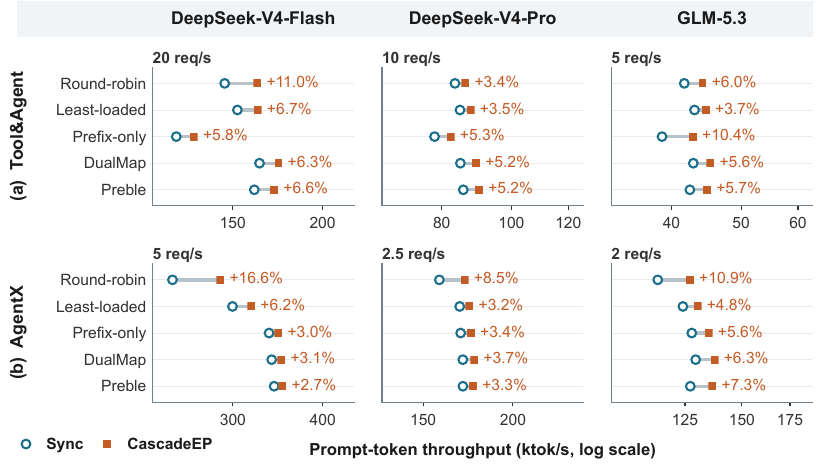}
\caption{Prompt-token throughput across models and workloads.}
\label{fig:throughput}
\end{figure}

\subsection{End-to-end prefill performance}
\label{sec:e2e}
Figure~\figref{fig:e2e} compares the p95 TTFT of \textsc{Sync} and \name{} across request rates for three models, two workloads, and five request-scheduling policies. At the highest request rate in each panel, \name{} achieves an average p95 TTFT speedup of 1.21$\times$ and a maximum of 1.48$\times$ over \textsc{Sync}. The largest reduction, 32.4\%, occurs on DeepSeek-V4-Flash with AgentX under least-loaded scheduling at 6~req/s. On AgentX, reductions at the highest tested rates average 22.4\% and 25.1\% across the three models under round-robin and least-loaded scheduling, respectively, compared with 7.0\%, 8.1\%, and 12.4\% under prefix-only, DualMap, and Preble. Although the gains are smaller under the latter policies, the improvements under DualMap and Preble demonstrate that intra-layer execution optimization complements cache- and load-aware request scheduling.

Figure~\figref{fig:throughput} compares prompt-token throughput at selected high-load request rates. \name{} improves throughput across all three models, both workloads, and all five policies, with an average increase of 6.0\% and a maximum of 16.6\% (1.06$\times$ and 1.17$\times$ the baseline throughput, respectively). The largest gain occurs on DeepSeek-V4-Flash with AgentX under round-robin scheduling at 5~req/s. Averaged across workloads and policies, throughput increases by 6.8\%, 4.5\%, and 6.6\% for DeepSeek-V4-Flash, DeepSeek-V4-Pro, and GLM-5.3, respectively, showing consistent throughput improvements alongside the TTFT reductions.

\subsection{Deep Dive}
\label{sec:ablation}
\paragraph{Ablation study.}
We compare four variants on DeepSeek-V4-Flash using AgentX at 5~req/s under round-robin placement, keeping all other settings fixed. \textsc{Sync} uses synchronous DEP execution. \textsc{Async} executes each source replica's ready expert inputs in a separate FFN GEMM kernel on a single compute stream, without a launch threshold. \textsc{Async+SF} adds \textit{streamFFN}'s launch threshold and multi-stream execution. \name{} further enables \textit{OEWF} to fetch expert weights and reassign eligible unstarted FFN work.

Figure~\figref{fig:ablation} shows the incremental benefits of these mechanisms. \textsc{Async} starts FFN computation earlier than \textsc{Sync}, but small FFN GEMM kernels limit its gains. Adding \textit{streamFFN} improves computation efficiency through threshold-based batching and multi-stream execution and yields the largest incremental TTFT reduction. The matched-work comparison below examines this efficiency improvement separately. On top of \textsc{Async+SF}, \textit{OEWF} uses idle GPU time to execute reassigned, unstarted FFN work, further reducing average and p95 TTFT by 7.8\% and 6.2\%, respectively.

\begin{figure}[!htbp]
\centering
\begin{minipage}[t]{0.48\linewidth}
\vspace{0pt}
\centering
\includegraphics[width=\linewidth]{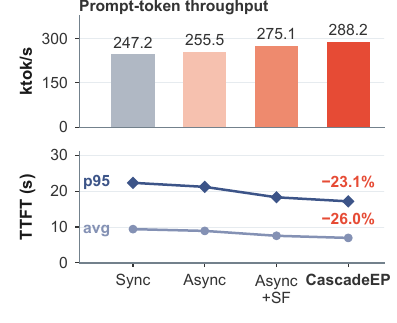}
\caption{Incremental ablation of \name{}. Top: prompt-token throughput. Bottom: average and p95 TTFT, annotated with the reduction of \name{} over \textsc{Sync}.}
\label{fig:ablation}
\end{minipage}\hfill
\begin{minipage}[t]{0.48\linewidth}
\vspace{0pt}
\centering
\includegraphics[width=\linewidth]{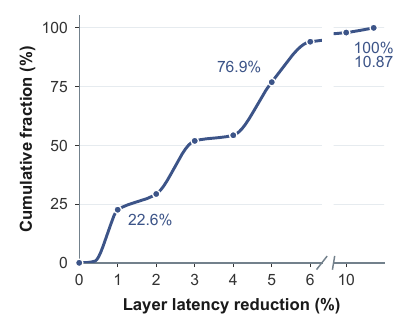}
\caption{CDF of per-layer latency reduction from enabling \textit{OEWF} for DeepSeek-V4-Flash on AgentX, with \textit{streamFFN} enabled in both configurations.}
\label{fig:oewf-layer}
\end{minipage}
\end{figure}

\paragraph{FFN computation efficiency.}
\label{sec:ffn-efficiency}
To isolate FFN computation efficiency from the benefits of early execution, we compare \textit{streamFFN} with synchronous EP's complete FFN GEMM kernels under identical expert weights and per-expert token counts. On DeepSeek-V4-Flash with AgentX, \textit{streamFFN} retains 97.3\% of synchronous EP's effective FFN throughput, incurring only a 2.7\% loss. The main reason is that the launch threshold preserves near-peak GEMM efficiency by avoiding overly small kernels, while multi-stream execution reduces GPU underutilization caused by the tail effect.

\paragraph{Layer-level acceleration.}
\label{sec:oewf-layer}
We analyze \textit{OEWF} by comparing execution with and without it for DeepSeek-V4-Flash on AgentX, with \textit{streamFFN} enabled in both configurations. Figure~\figref{fig:oewf-layer} shows that \textit{OEWF} reduces latency by at least 1\% in 77.4\% of layer forwards and by at least 5\% in 23.1\%, with a maximum observed reduction of 10.87\%. The main reason is that \textit{OEWF} uses otherwise idle GPUs to execute part of the FFN work assigned to replicas with longer attention times, reducing their remaining workload and shortening the layer tail.

%% file: sections/conclusion.tex
\Needspace{11\baselineskip}
\section{Conclusion}
This paper presented \name{}, an intra-layer execution engine that mitigates the impact of attention imbalance on MoE prefill. Asynchronous EP starts expert computation before tokens from all attention replicas arrive. \textit{streamFFN} combines kernel launch threshold and multi-stream execution to balance early launches with FFN efficiency, while \textit{OEWF} fetches expert weights and reassigns unstarted FFN work to use GPU idle time. We evaluate \name{} on DeepSeek-V4-Flash, DeepSeek-V4-Pro \citep{deepseek2026v4}, and GLM-5.3 \citep{zai2026glm53} using Tool\&Agent \citep{qin2025mooncake} and AgentX under five request-scheduling policies. The results show that \name{} achieves 1.13$\times$ speedup in p95 time-to-first-token (TTFT) on average (up to 1.48$\times$) and improves inference throughput by 1.06$\times$ on average (up to 1.17$\times$). These results show that intra-layer asynchronous execution complements request scheduling in improving MoE prefill performance.

%% file: sections/reproducibility_statement.tex
\subsection*{Reproducibility Statement}
Section~\ref{sec:evaluation} describes the implementation, experimental setup, workload replay procedure, baselines, and evaluation metrics. Appendix~\ref{app:dep-prefill} defines the timing boundaries and aggregation methods for attention-imbalance measurements, and Appendix~\ref{app:runtime-impl} details the runtime implementation. Appendix~\ref{app:fetch-proof} states the assumptions and provides the proof of \textit{OEWF}'s conditional non-degradation result. Appendix~\ref{app:estimate-accuracy} reports the accuracy of the runtime estimates used for donor selection.

%% file: sections/ai_statement.tex
\subsection*{AI Use Statement}
We used generative AI tools as assistants throughout this work, under author supervision. First, we used LLM assistance to aid and polish the writing; the authors wrote and verified every technical claim, number, and conclusion. Second, we used LLM-based coding assistants for research execution, to help implement parts of the serving system and the experiment scripts and to help debug and run experiments; the authors reviewed the resulting code and validated all reported results. We also used LLM assistance for retrieval and discovery, including finding related work; the authors checked the suggested literature and verified every citation. We did not use generative AI to generate synthetic datasets, to prove mathematical claims, or to fabricate data, citations, or experimental results.

%% file: appendices/dep_prefill.tex
\section{Long-Tail Statistics of DEP Prefill}
\label{app:dep-prefill}
We profile synchronous DEP prefill with DeepSeek-V4-Flash on the eight-GPU prefill instance described in \secref{sec:setup}. Mooncake's Tool\&Agent trace and AgentX are evaluated separately. We measure cross-replica differences in attention-stage duration and in the times at which replicas enter token dispatch.

\paragraph{Timing and readiness gaps.}
We use NVIDIA Nsight Systems to extract attention-stage durations and GPU-side dispatch-arrival timestamps from per-replica execution timelines. SGLang's scheduler coordinates each model forward through a DP metadata all-gather; we match layer-forward invocations across replicas by the coordinated forward index and layer index. This coordination identifies the same execution round but does not imply simultaneous attention starts across GPUs. In the synchronous DeepEP V2 intranode path used for profiling, replicas enter dispatch independently after local attention, routing, and input preparation. For layer-forward invocation $i$, we define $r_{d,i}$ as the GPU start timestamp of replica $d$'s \texttt{dispatch\_impl} kernel. The kernel performs cross-rank coordination internally, so its start timestamp precedes that coordination rather than marking the start of payload transfer \citep{deepep2025}. We define the cross-replica readiness spread as $s_i=\max_d r_{d,i}-\min_d r_{d,i}$ and each replica's readiness gap relative to the last-arriving replica as $w_{d,i}=\max_j r_{j,i}-r_{d,i}$. We use $w_{d,i}$ to measure replica idle time while waiting for other replicas to reach dispatch: the interval from replica $d$'s entry into dispatch until the last replica enters dispatch, during which synchronous FFN computation cannot begin. This measures waiting for replica arrival, rather than the full dispatch duration.

\paragraph{Normalization and aggregation.}
Let $T_{\mathrm{E2E}}$ denote the wall-clock interval from issuance of the first measured request until the prefill instance produces the first token for the last measured request to complete. Let $\mathcal I$ denote the measured layer-forward invocations within this same interval. We define the long-tail ratio as
\begin{equation}
R_{\mathrm{tail}}=\frac{\sum_{i\in\mathcal I}\sum_{d=1}^{D}w_{d,i}}{D T_{\mathrm{E2E}}},
\label{eq:readiness-gap-ratio}
\end{equation}
where $D=8$ is the number of attention replicas. This ratio normalizes the accumulated idle time while waiting for other replicas to reach dispatch, averaged across replicas, by the run's elapsed time. We compute it separately for each workload.


The long-tail ratio is 28.00\% on Mooncake's Tool\&Agent trace and 28.86\% on AgentX. On both workloads, idle time while waiting for other replicas to reach dispatch therefore averages more than one quarter of the end-to-end runtime per replica.


\paragraph{DP attention load balance.}
The attention stage includes query and key--value projections, attention-related normalization and positional encoding, local KV-cache updates, core attention computation, and output projection. For DeepSeek-V4-Flash, it also includes KV compression and, where applicable, sparse indexing \citep{deepseek2026v4}.

For each layer-forward invocation $i$, let $a_{d,i}$ denote the attention-stage duration of replica $d$, measured using identical stage boundaries across replicas. We define attention load balance (ALB) as
\begin{equation}
\mathrm{ALB}_{i}
=\frac{D^{-1}\sum_{d=1}^{D}a_{d,i}}
{\max_{1\le d\le D}a_{d,i}}\times100\%.
\label{eq:dp-attention-balance}
\end{equation}
Higher ALB indicates more balanced attention-stage execution, with 100\% corresponding to equal durations across replicas. For Figure~\figref{fig:motivation}, we average the per-invocation ALB scores, excluding invocations in which all replicas have zero attention-stage duration. ALB measures relative imbalance within each invocation, whereas the long-tail ratio measures accumulated readiness gaps relative to end-to-end runtime.

%% file: appendices/dsv4_attention.tex
\section{Attention Execution Characteristics}
\label{app:profiles}
Two attention replicas can receive the same number of requests but different amounts of prefill work, because requests can differ in prompt length and reusable KV-cache prefixes. We compare attention-stage execution times at a fixed total token count $N=BL$ across batch sizes $B$ and input lengths $L$, using attention-module profiles of DeepSeek-V4-Flash, GLM-5.3, and DeepSeek-V3-0324 \citep{deepseek2024v3}. DeepSeek-V3-0324 is included only in this comparison, not in the end-to-end evaluation.

\paragraph{Token count and context length.}
At fixed model dimensions, projection GEMM FLOPs scale with the number of processed tokens $N$. Context-dependent attention work need not follow the same scaling. For $B$ equal-length sequences of length $L$ with no reusable prefix, dense causal attention contains
\begin{equation}
P(B,L)=\frac{B L(L+1)}{2}=\frac{N(L+1)}{2},\qquad N=B L,
\label{eq:attention-pair-count}
\end{equation}
query--key pairs. At fixed head dimensions, attention-score and value-aggregation FLOPs scale with $P(B,L)$. Thus, at the same total token count, fewer long sequences require more dense-attention arithmetic than more short sequences. Sparse selection and KV compression reduce the context-dependent work. These reductions can make token-proportional operations more prominent without making the full attention stage independent of batch composition.

\paragraph{Attention timing distributions.}
DeepSeek-V4-Flash, DeepSeek-V3-0324 and GLM-5.3 are profiled as complete attention-module calls on eight B200 GPUs. The measurements use no prefix-cache reuse. Each panel of Figure~\figref{fig:attention-stage-boxplots} is normalized by its largest displayed mean time. Box spread reflects differences across configurations rather than run-to-run variability. At a fixed total token count, DeepSeek-V3-0324 and GLM-5.3 exhibit substantial execution-time variation across batch-size and input-length combinations, whereas DeepSeek-V4-Flash shows a smaller relative spread in the displayed profiles. FLOPs explain the workload dependence but do not directly determine execution time. Total token count is therefore a useful workload indicator, but does not guarantee equal attention-stage execution times across different batch compositions.

\begin{figure}[htbp]
\centering
\includegraphics[width=\linewidth]{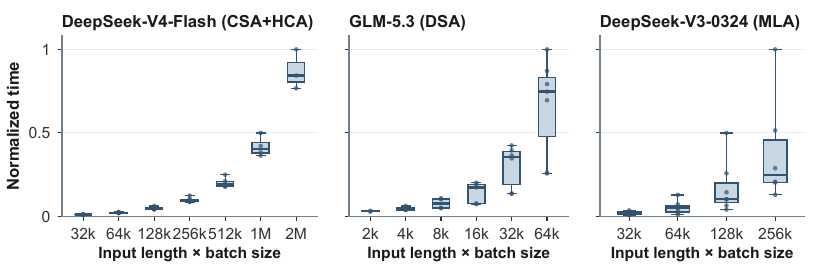}
\caption{Attention execution times across batch-size and input-length configurations at fixed total token counts, without prefix-cache reuse. Each panel is normalized by its largest displayed time.}
\label{fig:attention-stage-boxplots}
\end{figure}

\paragraph{Online batching and cache locality.}
Request placement selects a replica, whereas the local scheduler forms prefill batches from queued requests under token and memory budgets \citep{zheng2024sglang}. A request's batch is therefore not fixed by placement alone. Cache locality further constrains the choice: assigning a request to a less-loaded replica that lacks its reusable prefix requires additional KV recomputation or transfer. Cache-aware policies account for this trade-off \citep{srivatsa2025preble}, but must do so as queues and cache contents change. These constraints make it difficult to consistently balance attention-stage execution times across replicas through request placement alone.

\paragraph{Estimating the remaining attention time.}
With the replica placement fixed by the scheduling policy and the cache hits known, we estimate the remaining attention time. The remaining attention-stage time in Equation~\ref{eq:attention-remaining} is
\begin{equation*}
A_{d}
= \bigl[f(I_d, P_d)
- t_{d}^{\mathrm{att}}\bigr],
\end{equation*}
where $f(I_d, P_d)$ is the predicted duration of replica $d$'s complete attention stage, including projections, KV compression and sparse indexing where applicable, core attention computation, and the output projection, and $t_{d}^{\mathrm{att}}$ is the elapsed attention time. If $f(I_d,P_d)-t_{d}^{\mathrm{att}}$ is negative, we set $A_d$ to $0$. Replica $d$ is then excluded from the candidate donors. Let $\mathcal R_d$ denote the requests in replica $d$'s current batch, so that $I_d=\{I_i\}_{i\in\mathcal R_d}$ and $P_d=\{P_i\}_{i\in\mathcal R_d}$. Request $i$ computes $q_i=I_i-P_i$ new tokens, and its $j$-th new token attends to at most $\min(k_i,P_i+j)$ keys, where $k_i$ is the per-query key budget of sparse selection and $k_i=\infty$ for dense attention. Following prefill cost models that aggregate per-request token counts and squared lengths \citep{zhong2024distserve} and account for prefix-cache hits \citep{qin2025mooncake}, we accumulate two workload features over the requests in the batch:
\begin{equation}
Q_d=\sum_{i\in\mathcal R_d} q_i,
\qquad
\Pi_d=\sum_{i\in\mathcal R_d}\sum_{j=1}^{q_i}\min\bigl(k_i,\,P_i+j\bigr),
\label{eq:attention-features}
\end{equation}
and predict the attention-stage duration as
\begin{equation}
f(I_d,P_d)=\alpha_0+\alpha_1 Q_d+\alpha_2\Pi_d.
\label{eq:attention-cost-model}
\end{equation}
$Q_d$ accounts for token-proportional operations such as the projections, and $\Pi_d$ counts the query--key pairs of core attention. For $B$ equal-length requests of length $L$ with dense attention and no reusable prefix, $Q_d=N$ and $\Pi_d=P(B,L)$ in Equation~\ref{eq:attention-pair-count}; $\Pi_d$ generalizes this pair count to heterogeneous lengths, matched prefixes, and sparse selection. Because both features are summed over requests, length variation within a batch is retained rather than averaged away. For each deployed model, GPU, and attention configuration, we fit $\alpha_0$, $\alpha_1$, and $\alpha_2$ to offline profiles by weighted least squares that minimizes relative error. Layers with the same attention configuration share one set of coefficients; in DeepSeek-V4-Flash, for example, layers with different KV-compression settings are fitted separately. The profiles cover batch sizes $B\in\{1,2,4,8,16,32,64\}$ and input lengths $I\in\{\text{1K},\text{2K},\text{4K},\text{8K},\text{16K},\text{32K},\text{64K}\}$ tokens, omitting infeasible combinations that exceed GPU memory. For these equal-length profiles, the matched prefix length is $P_i=rI$, with hit ratio $r\in\{0\%,20\%,40\%,60\%,70\%,80\%,90\%,95\%\}$. The linear form does not model kernel tiling or wave-quantization effects; Appendix~\ref{app:estimate-accuracy} reports the measured error.

%% file: appendices/runtime_implementation.tex
\section{Runtime Implementation Details}
\label{app:runtime-impl}

\paragraph{Asynchronous dispatch.}
Our DeepEP extension maintains a separate communication lane for each source--destination pair. At the destination, a lane becomes ready after the source's routed activations and associated metadata have been received and packed into expert-major order. A generation-tagged completion signal identifies the corresponding dispatch instance, while system-scope release/acquire ordering ensures that the data are visible before consumption. The SGLang adapter uses CUDA stream waits on these signals, allowing inputs from one attention replica to become eligible for computation without waiting for the other replicas. Readiness does not immediately trigger a grouped FFN invocation: \textit{streamFFN} accumulates ready work and applies the launch threshold and final-flush rule described in \secref{sec:sgffn}.

\paragraph{Grouped FFN invocations.}
For each grouped FFN invocation, a GPU-side Triton layout transform combines the selected rows from ready source lanes into contiguous groups by expert. It computes per-expert row offsets from the selected token counts, aligns each group to the GEMM tile size, and packs the activations together with their quantization metadata. This allows rows for the same expert to share one compute group even when they arrive through different source lanes. An inverse map preserves each row's source replica, token index, and routing slot for returning and combining the outputs. Within an invocation, the two expert GEMMs and the intervening activation and quantization operations execute in stream order. Independent invocations use separate CUDA streams and scratch buffers, preserving these dependencies while permitting concurrent execution.

\paragraph{One-shot weight fetching.}
\textit{OEWF} uses a preallocated temporary weight buffer on each helper to store the fetched expert matrices and quantization metadata. A per-rank, per-layer-forward flag permits one fetch decision after local dispatch completes. Following the selection rule in \secref{sec:swf}, the coordinator chooses the eligible donor with the latest predicted FFN completion time and fixes the expert and token sets within a byte budget bounded by the buffer's available capacity. The implementation considers only tokens dispatched by the helper whose inputs it retains locally. Before issuing asynchronous copies, the coordinator validates shared reservations and reserves the protected transfer service required by the design. The selected experts' weights are copied in full from the donor into this buffer. During prefetching, the selected tokens remain in the donor's workload and retain their original launch eligibility; the donor does not wait for the copies. Weight readiness triggers a separate FFN reassignment decision. Empty selections, reservation failures, failed copies, and rejected FFN reassignments do not trigger another fetch selection.

\paragraph{FFN reassignment commit and fallback.}
Once weights are ready, the coordinator considers tokens recorded in the helper's fetch plan. A token is eligible only if its donor batch has a fixed token set and its FFN execution has not been submitted. It evaluates the admission test in Equation~\ref{eq:swf-admission} against the latest committed plan, including earlier helper commitments, and requires a protected interval that accommodates the complete helper invocation without delaying existing work. A version-checked atomic commit revalidates the plan, invocation state, and token ownership, then transfers execution ownership and updates the donor's residual token counts and the helper's reservation together. If validation fails, the donor has already submitted the work, or another helper has claimed it, the attempted FFN reassignment is abandoned and execution remains with the current owner. Successfully reassigned tokens execute on their retained local inputs at the helper. The donor retains its original launch eligibility even if reassignment reduces its remaining work below the launch threshold. The non-degradation guarantee remains conditional on the timing bounds and resource protections stated in Appendix~\ref{app:fetch-proof}.

\paragraph{Output return and combine.}
Expert outputs are returned asynchronously to their originating attention replicas, or passed directly to local combine when computed there. Source-token and routing-slot metadata associate each contribution with the corresponding token and selected expert. Each routed-expert contribution is weighted by its original routing coefficient exactly once. Combine returns and accumulates these contributions at the originating attention replica, and the result is added to the shared-expert output. The runtime tracks output readiness separately from buffer reclamation: buffers remain live until their dependent transfers and computations have completed. Although contributions can arrive and be accumulated independently, \name{} retains the common layer boundary. The next layer is admitted only after every attention replica has received its required MoE outputs.

%% file: appendices/estimate_accuracy.tex
\section{Runtime Estimates and Their Accuracy}
\label{app:estimate-accuracy}
Coordinated greedy expert-weight selection (\secref{sec:swf}) ranks candidate donors using two predicted quantities: the remaining attention-stage time $A_{d}$ in Equation~\ref{eq:attention-remaining} and the received expert-token count $N_{d}$ in Equation~\ref{eq:donor-completion}. Appendix~\ref{app:profiles} describes the attention-time estimator. This appendix describes how the received expert-token count is estimated and reports the accuracy of the complete attention-stage duration used to form $A_{d}$, together with the accuracy of $N_{d}$. The non-degradation analysis in Appendix~\ref{app:fetch-proof} assumes accurate timing; these measurements quantify the estimation error observed in practice and do not extend that guarantee.

\paragraph{Estimating received expert-token counts.}
Let $N_{\ell,d}$ denote the number of expert tokens that rank $d$ receives in MoE layer $\ell$. When the current layer's count is not yet known, we estimate it with the count that the same rank received in the preceding MoE layer:
\begin{equation}
N_d=N_{\ell-1,d}.
\label{eq:recv-token-estimate}
\end{equation}
Within a forward pass, every MoE layer routes the same tokens, each to the same number of experts, so the total number of token--expert pairs is identical across layers; only its distribution across ranks changes with routing. The estimate therefore approximates the current layer's per-rank total with that of the preceding layer. The first MoE layer has no preceding layer, and \textit{OEWF} skips it.

\paragraph{Error metric.}
For each sample, we measure the gap between the predicted value $\hat y$ and the measured value $y>0$ relative to the measured value, $\epsilon=|\hat y-y|/y$, and report accuracy as $1-\bar\epsilon$, where $\bar\epsilon$ is the mean of $\epsilon$ over all samples. Both estimates are validated online while serving DeepSeek-V4-Flash, DeepSeek-V4-Pro, and GLM-5.3 on the Tool\&Agent and AgentX workloads at the request rates in Figure~\figref{fig:e2e}.

\paragraph{Attention-stage time.}
For each replica and MoE layer, we record the input lengths and matched prefix lengths of the current batch together with the measured attention-stage duration, using the stage boundary defined for $f$ in Appendix~\ref{app:profiles}. The accuracy of 90.3\% compares $f(I_d,P_d)$ with this measured complete attention-stage duration.

\paragraph{Received expert-token count.}
For each rank and layer, we compare $N_d$ from Equation~\ref{eq:recv-token-estimate} with the number of expert tokens that rank $d$ actually receives in layer $\ell$. Because \textit{OEWF} skips the first MoE layer, these samples start from the second MoE layer. The estimate achieves an accuracy of 85.9\%.

%% file: appendices/analysis.tex
\section{Non-degradation of Opportunistic Expert Weight Fetching}
\label{app:analysis}
\label{app:fetch-proof}

This appendix shows how the FFN reassignment rule in \secref{sec:swf} gives non-increasing MoE layer completion time relative to execution without weight fetching, under the assumptions stated below. We first consider a donor--helper pair, then extend its completion-time comparison to all ranks and successive plan updates.

\begin{proposition}[From two ranks to any number of ranks]
\label{prop:fetch-certified}
Under the conditions stated in the proof, sequential updates satisfying Equation~\ref{eq:swf-admission} give $M_{\rm OEWF}\le M_{\rm no\text{-}fetch}$ for any number of ranks.
\end{proposition}
\begin{proof}
Fix the layer inputs and routing. At a weight-ready decision, let $T_r(S)$ denote rank $r$'s remaining FFN completion time under the current plan $S$, measured from the same decision instant for both plans. Assume the completion-time model in \secref{sec:swf}, including FFN latencies under concurrent execution and the overlap estimate, is accurate, and weight fetching does not delay the current plan. A reassignment preserves the donor's attention and dispatch times and the completion times of the helper's existing work. Using the same FFN latency function $F$ as Equation~\ref{eq:ffn-throughput-latency}, with $F(0)=0$, transferring $0\le N_{\mathrm{trans}}\le N_d$ unstarted expert tokens gives
\begin{equation}
\begin{aligned}
T_d(S) &= A_d+D_d+F(N_d),\\
T_d(S') &= A_d+D_d+F(N_d-N_{\mathrm{trans}}),\\
T_h(S') &= T_h(S)+\alpha_h F(N_{\mathrm{trans}}).
\end{aligned}
\label{eq:swf-proof-times}
\end{equation}
For $N_{\mathrm{trans}}>0$, let $s_h\in[0,T_h(S)]$ be the planned launch time of the reassigned FFN kernel, measured from this decision. Setting $\alpha_h=\max\{0,s_h+F(N_{\mathrm{trans}})-T_h(S)\}/F(N_{\mathrm{trans}})$ gives $\alpha_h\in[0,1]$ and $T_h(S')=\max\{T_h(S),s_h+F(N_{\mathrm{trans}})\}$, accounting for partial or complete overlap. With no transferred tokens, the plan is unchanged. Equation~\ref{eq:swf-admission} accepts the update only when
\begin{equation}
\max\{T_d(S'),T_h(S')\}
\le\max\{T_d(S),T_h(S)\}.
\label{eq:swf-local-reduction}
\end{equation}
Reassignment does not delay ranks outside $\{d,h\}$. Taking the maximum over all ranks therefore gives
\begin{equation}
T_{\max}(S')=\max_r T_r(S')\le\max_r T_r(S)=T_{\max}(S).
\label{eq:swf-global-reduction}
\end{equation}
Let $M(S)$ be the layer completion time measured from layer entry, $t$ the decision time on this same origin, and $L(S)=M(S)-t-T_{\max}(S)$ the residual return/combine and layer-synchronization tail after the last FFN completes. Assume this tail does not increase, including for helper outputs, so $L(S')\le L(S)$. Then
\begin{equation}
\begin{aligned}
M(S') &= t+T_{\max}(S')+L(S')\\
&\le t+T_{\max}(S)+L(S)=M(S).
\end{aligned}
\label{eq:swf-layer-reduction}
\end{equation}
A rejected update retains $S$. The coordinator applies each accepted update to the latest joint plan, so its token counts, completion times, and overlap estimates include all previous reassignments. Starting from the no-fetch plan $S_0$, whose execution is unaffected by fetching, induction gives
\begin{equation}
M_{\rm OEWF}=M(S_K)\le M(S_{K-1})\le\cdots\le M(S_0)=M_{\rm no\text{-}fetch}.
\label{eq:swf-chain}
\end{equation}
Here $S_K$ is the final \textit{OEWF} plan.
\end{proof}